\documentclass[sigconf,nonacm]{acmart}
\AtBeginDocument{%
  \providecommand\BibTeX{{%
    \normalfont B\kern-0.5em{\scshape i\kern-0.25em b}\kern-0.8em\TeX}}}

\usepackage{graphicx}
\usepackage{xcolor}
\usepackage{mathtools} 
\usepackage{subcaption}
\usepackage{hyperref}
\usepackage{algorithm}
\usepackage{algpseudocode}
\usepackage{booktabs}

\usepackage{bbm}
\usepackage{amsmath}

\DeclareMathOperator*{\argmin}{arg\,min}
\newcommand{\method}{\texttt{SHARE}}
\newcommand{\algName}{\texttt{STABLE}}
\begin{document}

\title{Identifying and Mitigating Instability in Embeddings of the Degenerate Core}

\author{David Liu}
\email{liu.davi@northeastern.edu}
\orcid{0000-0002-2129-447X}

\affiliation{%
  \institution{Northeastern University}
  \city{Boston}
  \state{MA}
  \country{USA}}
  
\author{Tina Eliassi-Rad}
\email{tina@eliassi.org}
\orcid{0000-0002-1892-1188}

\affiliation{%
  \institution{Northeastern University}
  \city{Boston}
  \state{MA}
  \country{USA}}

\renewcommand{\shortauthors}{Liu and Eliassi-Rad}

\begin{abstract}
  Are the embeddings of a graph's degenerate core stable? What happens to the embeddings of nodes in the degenerate core as we systematically remove periphery nodes (by repeated peeling off $k$-cores)? We discover three patterns w.r.t. instability in degenerate-core embeddings across a variety of popular graph embedding algorithms and datasets. We use regression to quantify the change point in graph embedding stability. Furthermore, we present the   \texttt{STABLE} algorithm, which takes an existing graph embedding algorithm and makes it stable. We show the effectiveness of \texttt{STABLE} in terms of making the degenerate-core embedding stable and still producing state-of-the-art link prediction performance.
\end{abstract}




\maketitle

\section{Introduction}
Previous work has presented varied evidence for the effectiveness of graph (a.k.a.~node) embedding algorithms. For instance, while some suggest that graph embeddings improve performance on link prediction and node classification, others have shown that basic heuristics can outperform graph embeddings in community detection~\cite{embedding-communities}. Other work has shown that the low dimensionality of embeddings prevents them from capturing the triangle structure of real-world networks~\cite{Seshadhri}. In this work, we examine the stability of graph embeddings as a means for better understanding the information they capture and their utility in different contexts. 

We measure the stability of the graph's degenerate core (i.e., its $k$-core with maximum $k$) as outer $k$-shells (i.e., the ``periphery") are iteratively shaved off. The $k$-core of an undirected graph $G$ is the maximal subgraph of $G$ in which every node is adjacent to at least $k$ nodes. A common approach to  understanding stability 
 is to measure changes to algorithmic output due to input perturbations. K-core analysis gives us a principled mechanism for changing graphs. In analyzing the embedded $k$-cores, we  ask whether the embeddings capture the degenerate core's structure, and if/how its embedding changes as each shell is removed. For example, previous work showed that degenerate cores are generally not cliques but contain community structure~\cite{corescope}. We study whether such patterns appear in the embeddings of the degenerate cores as $k$-shell are removed. We also investigate whether the stability of the graph's degenerate-core embedding varies by graph type (such as protein-protein interaction, social network, etc.), or varies by graph embedding algorithm (such as matrix factorization methods or skip-gram methods).

It is important to evaluate the stability of embeddings of nodes in the degenerate-core (a.k.a.~``degenerate-core embeddings'') because dense subgraphs are the ``heart" of the graph. Nodes in the degenerate core are often the most influential spreaders. In marketing applications, the removal of a dense core node can trigger a cascade of node removals \cite{liu_core-like_2015,anchors}. Yet, as important as the core nodes are, previous studies on online activism have shown that core nodes are also dependent on periphery nodes to amplify messages originating from the core nodes \cite{critical-periphery}. In this study, we assess the importance of the periphery nodes in the stability of the core node embeddings -- specifically, the nodes in the degenerate core. 

As we present in this work, the embedding of the nodes in the degenerate core are not stable (as in they do not persist as the periphery is removed). Thus, graph embeddings are relative and not absolute. These possible perturbations are a concern because real-world networks are noisy \cite{noisy_networks,noisy_networks_bayesian} and dynamic \cite{dynamic-embedding-survey}. As such, unstable embeddings should push us to place less faith in any individual set of graph embeddings. Instead, we must better specify the noise in the network to qualify the quality of the graph embedding.

Our main contributions are as follows:
\begin{enumerate}
    \item Across multiple categories of graphs and embedding algorithms, we discover three patterns of instability in embeddings of nodes in the degenerate core. In the process we introduce a method called \method{} for measuring the stability of degenerate-core embeddings. 
    \item We show that the instability in degenerate-core embeddings is correlated with increases in graph edge density when the periphery is removed. The correlation with density surpasses correlations with other graph properties, notably graph size, as the periphery is iteratively removed.
    \item We present an algorithm \algName{} for generating core-stable graph embeddings. Our algorithm is flexible enough to augment any existing graph embedding algorithm. We show that when instantiated with Laplacian Eigenmaps~\cite{lapeigenmap} and LINE~\cite{LINE}, our algorithm yields embeddings that preserve downstream utility while increasing stability.
\end{enumerate}
\section{Defining Embedding Instability}\label{sec:method}
We analyze the stability of graph  embedding algorithms by tracking the embedding of the degenerate core (the $k$-core with maximal $k$), as we progressively shave k-shells.\footnote{Our analysis pertains to undirected graphs without self-loops.} After each shell is removed, we re-embed the remaining subgraph; we name this method SHave-And-Re-Embed, or ``\method{}". Figure \ref{fig:shaving-method-example} illustrates the application of \method{} on the Zachary Karate Club graph. In the figure, the degenerate core ($k=4$) is highlighted in blue in the top row, and the bottom row plots the Node2Vec embedding for the remaining subgraph, where the corresponding degenerate-core embeddings are also plotted in blue. For the Karate Club graph, the relative proximities of the core embeddings do not change until the degenerate core is embedded in isolation, which we call the \emph{isolated embeddings}.
\begin{figure}
    \centering
    \includegraphics[width=\linewidth]{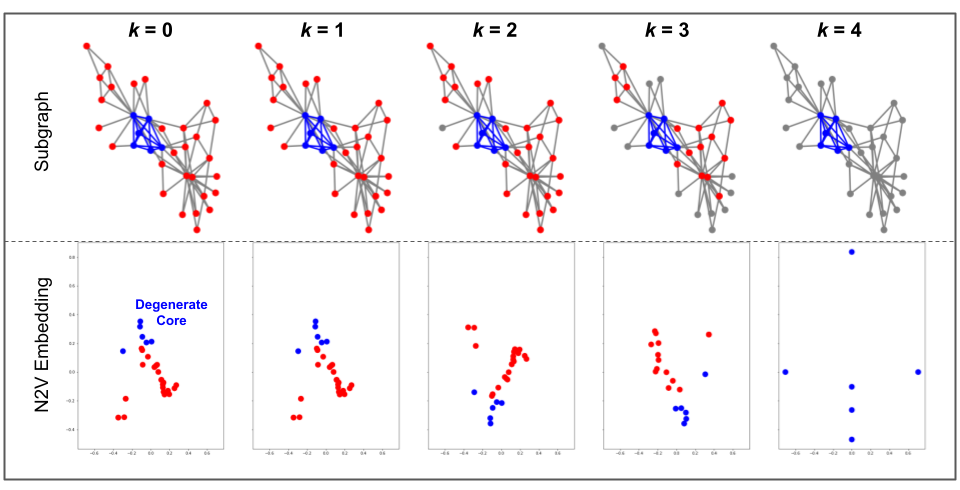}
    \caption{To measure the stability of degenerate-core embeddings, we use a method that we call SHave-And-Re-Embed, or ``\method{}". In this figure, we apply \method{} to the Zachary Karate Club graph~\cite{zachary}. At each iteration, we shave off the outermost k-shell (top row) and re-embed the remaining subgraph (bottom row). We then analyze the stability of the degenerate-core embeddings. Removed nodes are shown in grey, remaining subgraph nodes in red, and degenerate-core nodes in blue. The Node2Vec embedding plots all share the same scale. The embeddings of the degenerate core vary widely as the periphery is removed.}
    \label{fig:shaving-method-example}
\end{figure}

We define stability as the property of being resilient to perturbation, a definition of stability that is common in the complex networks literature \cite{Stability_Democracy}. This definition contrasts stability with robustness, which is an insensitivity to microscopic changes across different settings or environments. In the case of core embeddings, we define stability as the resilience of the proximities of degenerate core embeddings when the periphery is perturbed. 

To quantify the stability of degenerate-core embeddings, we use \method{} to measure the evolution of the distribution of pairwise distances in the embedded space, which we call the \emph{degenerate-core pairwise distribution}. Specifically, \method{} begins with the entire graph ($k=0$), embeds the nodes and calculates the distribution of pairwise 
distances among the embeddings for the degenerate-core nodes. Next, \method{} takes the $k=1$ core, re-embeds the nodes in the subgraph, and re-calculates the pairwise distribution for the degenerate-core nodes using the updated embeddings. If the embeddings are stable, the pairwise distributions should not vary as the $k$-shells are removed. The advantage of using the pairwise distribution is that it captures the relative geometric relationships among the embeddings as opposed to the absolute positions in the embedding space. For instance, in Figure \ref{fig:shaving-method-example}, the embeddings invert after shaving the $k=1$ shell, nevertheless the relative distances among embeddings are largely unchanged. An alternative measure of stability would be the Frobenius norm of the difference in weighted adjacency matrices \cite{Dyngem}; however, we found that the pairwise  distribution provides a more granular measure of stability as opposed to a single norm value. After calculating the degenerate-core pairwise distribution at each $k$, \method{} measures the distance among the distributions with the Earth Mover's Distance (EMD). 

We define \textit{instability} ($\Delta_k$) at core $k$ as the increase in EMD relative to the original graph, where $D_k$ is the degenerate-core pairwise distribution for the $k$-core. Table \ref{tab:notation} summarizes our notation. 
\begin{equation}\label{eqn:def-instability}
   \textrm{Instability }  \Delta_{k} = \text{EMD}\left(D_k, D_0\right) - \text{EMD}\left(D_{k-1}, D_0\right)
\end{equation}
\begin{table}[ht]
    \centering
    \begin{tabular}{| p{0.15\linewidth} | p{0.75\linewidth} | }
        \hline
         Symbol &  Meaning \\ \hline
         $n,m$ & number of vertices, edges\\
         $C$ & number of connected components \\
         $k_{\text{max}}$ & degeneracy of a graph\\
         $d$ & embedding dimension\\
         $D_k$ & the degenerate-core pairwise distribution for $k$-core $k$. See Sec. \ref{sec:method}\\
         $\Delta_k$ & core instability for the $k$-core. See Eqn. \ref{eqn:def-instability}.\\
         $w_{ij}$ & weight of edge $\{i, j\}$ \\
         $W_S$ & adjacency matrix for subgraph induced by a set of nodes $S$\\
         $\mathcal{D}$ & set of nodes in the degenerate core\\
         $\pmb{u}_i$ & $d$-dimensional embedding vector for node $i$\\
         $\pmb{Y}$ & $n \times d$ matrix containing all of the node embeddings where row $i$ is $\pmb{u}_i^T$\\
         $\mathcal{L}_b, \mathcal{L}_s$ & base and stability loss functions\\
         $\alpha, \beta, \gamma$ & regularization hyperparameters\\
         $n_b$ & number of training batches\\
         $\eta$ & learning rate\\
        \hline
    \end{tabular}
    \caption{Notations used in this paper.}
    \label{tab:notation}
\end{table}

\subsection{Graph Embedding Algorithms and  Datasets}
\label{sec:algs_and_graphs}
We ran the proposed stability analysis using a combination of graph embedding algorithms---namely, HOPE \cite{HOPE}, Laplacian Eigenmaps \cite{lapeigenmap}, Node2Vec \cite{node2vec}, SDNE \cite{SDNE}, Hyperbolic GCN (HGCN) \cite{HGCN}, and PCA as a baseline. We picked these graph embedding algorithms because they span the taxonomy provide in Chami et al.~\cite{taxonomy}. To find instability patterns, we experimented on vary of graph datasets (listed in Table \ref{tab:snap_datasets}). The datasets are primarily collected from SNAP \cite{snapnets}, with the exception of Wikipedia \cite{wiki}, Autonomous Systems (AS) \cite{autonomous-systems}, and the synthetic graphs, which were generated to be of similar size as the real-world graphs.
\begin{table}[ht]
    \centering
    \begin{tabular}{| l | l | r | r | r || r | r |}\hline
         \textbf{Graph} & \textbf{Type} & $n$ & $m$ & $C$ & $k_{\text{max}}$ & \textbf{\% in $\mathcal{D}$}\\ \hline
         Wikipedia & Language & 4.8K & 185K & 1 & 49 & 3.1\%\\
         Facebook & Social & 4.0K & 88K& 1 & 115 & 3.9\%\\
         PPI & Biological & 3.9K & 77K & 35 & 29 & 2.8\%\\
         ca-HepTh & Citation & 9.9K & 26K & 429 & 31 & 0.3\%\\
         LastFM & Social & 7.6K & 28K & 1 & 20 & 0.6\%\\ 
         AS & Internet & 23K & 48K & 1 & 25 & 0.3\%\\ \hline
         ER ($p = 2e^{-3}$) & Synthetic & 5K & 25K & 1 & 7 & 67\%\\
         ER ($p = 4e^{-3}$) & Synthetic & 5K & 50K & 1 & 14 & 87\%\\
         BA ($m=5$) & Synthetic & 5K & 25K & 1 & 5 & 100\%\\
         BA ($m=10$) & Synthetic & 5K & 50K & 1 & 10 & 100\%\\
         BTER (PA) & Synthetic & 5K & 25K & 1 & 11 & 1.5\%\\
         BTER (Arb.) & Synthetic & 4.8K & 35K & 314 & 51 & 3.3\%\\\hline
    \end{tabular}
    \caption{Graph datasets used in our study. \% in $\mathcal{D}$ refers to the percentage of nodes in the degenerate core. ER, BA, and BTER are short for Erd\"{o}s-R\'{e}nyi~\cite{er}, Barab\'{a}si-Albert~\cite{BA}, and Block Two-Level Erd\"{o}s-R\'{e}nyi~\cite{bter} random graphs, respectively. PA is for a degree distribution that exhibits preferential attachment. Arb.~is for an arbitrary degree distribution. A description of the synthetic graphs is in Appendix \ref{sec:synthetic-network}. }
    \label{tab:snap_datasets}
\end{table}
Figure \ref{fig:datasets} shows that the selected graphs span a variety of k-core structures, as defined by the graph degeneracy and the maximum-core link entropy~\cite{liu_core-like_2015}, which is high when the degenerate core is well-connected with the outer shells.
\begin{figure}
    \centering
    \includegraphics[width=0.9\linewidth]{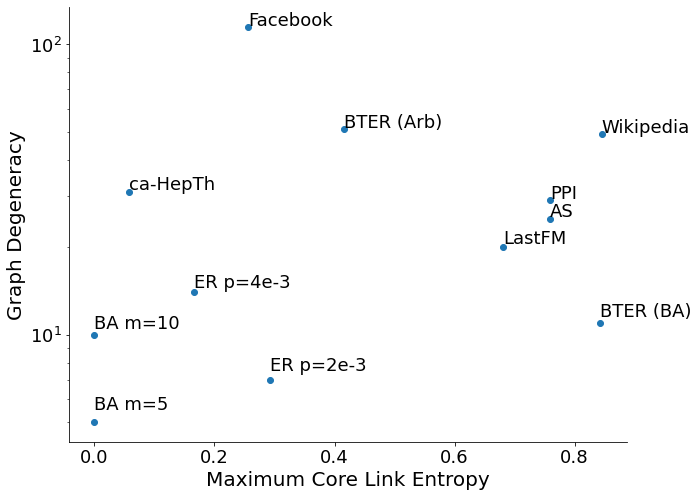}
    \caption{We chose graph datasets that have diverse $k$-core structures. Each graph is plotted based on its degeneracy vs. its maximum-core link entropy ($\in [0,1]$).  Liu et al. \cite{liu_core-like_2015} define maximum-core link entropy, where higher values correspond with degenerate cores that are well-connected with the outer shells; and graphs with low maximum-core link entropy have degenerate cores isolated from the rest of the graph. The Y-axis is the graph's degeneracy (the largest $k$ in the graph's $k$-cores).}
    \label{fig:datasets}
\end{figure}
\section{Unstable Degenerate-Core Embeddings}
We provide our results analyzing the stability of degenerate-core embeddings in three parts. First, we analyze the extreme case for matrix factorization (see Sec. \ref{sec:embed-cliques}) and skip-gram based graph embedding algorithms (Appendix \ref{sec:n2v_clique}) and show that they generate arbitrary embeddings. Second, we show that as the periphery $k$-shells are removed, the evolution of degenerate-core embeddings follows three patterns across the various graph types and graph embedding algorithms. Third, we show that points of embedding instability are correlated with increases in the subgraph density.
\subsection{Embedding Cliques}\label{sec:embed-cliques}
To gain an intuition of the stability of degenerate-core embeddings, we begin by examining the extreme case: a clique (i.e., a graph with maximum density) embedded in isolation (without periphery). As a case study, let us examine how Laplacian Eigenmaps \cite{lapeigenmap} embeds cliques.\footnote{For an empirical analysis of Node2Vec clique embeddings, see Appendix \ref{sec:n2v_clique}.} At a high level, graph embedding algorithms embed similar nodes in the graph space closer to each other in the embedding space. Similarity in the graph space has multiple definitions. Two popular definitions are: (i) nodes $u$ and $v$ are similar if the two are neighbors \cite{lapeigenmap} (ii) node $v$ is similar to $u$ if $v$ appears in a random walk starting at $u$ \cite{node2vec}. With cliques, all nodes are connected to each other and are thus equally similar. 

Laplacian Eigenmaps embeddings are based on the eigenvectors of the random-walk normalized graph Laplacian. For a clique of $n$ nodes, the random-walk normalized Laplacian ($L_{rw}$) is as follows:
\begin{align}
    L_{rw} &= D^{-1}L\\
    &= D^{-1}(D - A)\\
    &= \frac{D - A}{n-1}\\
    &= \begin{bmatrix}
    1 & -\frac{1}{n-1}\\
    -\frac{1}{n-1}& 1\\
    \end{bmatrix}
\end{align}
Where $D$ is the degree matrix; $L$ is the graph Laplacian; and $A$ is the adjacency matrix. The random-walk normalized Laplacian is a matrix in which all diagonal entries are $1$ and all off-diagonal entries are $-\frac{1}{n-1}$.  The eigenvalues for the above normalized Laplacian matrix are characterized in Theorem \ref{theorem:clique-lap-eigenvalues}, the proof of which is in Appendix \ref{sec:embed-clique}.  
\begin{theorem}\label{theorem:clique-lap-eigenvalues}
The random-walk normalized Laplacian for a clique of size $n$ has two eigenvalues: zero, with multiplicity one, and $1 + \frac{1}{n-1}$ with multiplicity $n-1$.
\end{theorem}

The eigenvalue of zero is evident because the clique is a single connected component. The equal non-zero eigenvalues of the clique Laplacian illusrate the arbitrary nature of Laplacian Eigenmaps embeddings for cliques. Laplacian Eigenmaps  assembles the embeddings by taking the eigenvectors corresponding to the $d<n$ smallest eigenvalues, after dropping the smallest eigenvalue. However, because all of the non-zero eigenvalues are equal, the eigenvectors chosen are an arbitary subset of the $n-1$ orthogonal eigenvectors. It is also worth noting that existing methods to determine the optimal embedding dimension $d$ by locating the elbow point in the loss function would fail to find an optimal threshold \cite{gu_principled_2021}. 

While the degenerate cores of the real-world graphs studied are not cliques and have more community structure than cliques, several are quite close to being complete graphs. Table \ref{tab:maxcore_completeness} lists the graph completeness of the various real-world degenerate cores studied, where given a degenerate core $G_\mathcal{D}=(V_\mathcal{D}, E_\mathcal{D})$ the degenerate core completeness is $\frac{\lvert E_\mathcal{D}\lvert}{\binom{\lvert V_\mathcal{D}\lvert}{2}}$. For instance, the core for the ca-HepTh citation network is indeed a clique and the core for the Facebook graph has a completeness of $0.898$. Thus, our analysis of the instability of clique embeddings is applicable to the degenerate cores.
\begin{table}[ht]
    \centering
    \begin{tabular}{| l | r | r |}\hline
         Graph &  Degeneracy & Degenerate-Core Completeness\\\hline
         ca-HepTh & 31 & 1.0\\
         Facebook & 115 & 0.898\\
         LastFM & 20 & 0.614\\
         AS & 25 & 0.545\\
         Wikipedia & 49 & 0.526\\
         PPI & 29 & 0.404\\\hline
    \end{tabular}
    \caption{Analyzing the behavior of graph embedding algorithms is useful because the degenerate cores for multiple real-world graphs resemble cliques. The degenerate (max) cores for the real-world graphs were extremely dense and several (Facebook and ca-HepTh) were nearly cliques. Given a degenerate core $G_\mathcal{D}=(V_\mathcal{D}, E_\mathcal{D})$, the degenerate core completeness is $\frac{\lvert E_\mathcal{D}\lvert}{\binom{\lvert V_\mathcal{D}\lvert}{2}}$.}
    \label{tab:maxcore_completeness}
\end{table}
\subsection{Patterns in Stability of Degenerate-Core Structural Representation}
After running \method{} (our degenerate core stability method described in Section \ref{sec:method}) on the 12 graphs and 6 embedding algorithms (described in Section \ref{sec:algs_and_graphs}), we observed the following three patterns.

\vspace*{3pt}
\textbf{Pattern 1: For many graph data, embedding algorithm combinations, the distribution of pairwise distances among degenerate core nodes shifts after removing a specific k-shell, but is stable otherwise.}
\vspace*{3pt}

 We observe that not only does the pairwise distribution change as k-shells are removed but the change often occurs abruptly. In Figure \ref{fig:shift-lastfm}, we show the degenerate-core pairwise distribution for the LastFM graph when embedded with HOPE \cite{HOPE}. For readability in all distance distribution figures (\ref{fig:shift-lastfm}-\ref{fig:unimodal}), we plot the distribution at intervals of $k$. The distinguishing feature is that the distributions for $k=1$ and $k=9$ are quite similar (left-skwed). However, the distribution for $k=13$ differs dramatically; then for $k > 13$, the distribution remains quite similar. This pattern suggests that nodes with coreness $k\in[9,13]$ greatly affect the embedding of the degenerate core ($k=20$). The exact critical shell will be identified in Section \ref{sec:casestudy}.

The $k$ at which the degenerate-core pairwise distribution transforms during the $k$-core shaving process depends on the graph and the embedding dimension. For instance, Figure \ref{fig:shift-facebook} shows the pairwise distribution for the Facebook graph using Laplacian Eigenmaps \cite{lapeigenmap}. When embedding with the entire graph, the degenerate core nodes are all very tightly embedded together, with the pairwise distance distribution concentrated near zero. As shells are removed the distribution flattens. However, the point ($k$) at which the distribution flattens depends on $d$. When $d=10$, the distribution flattens after $k=47$. On the otherhand, when $d=20$, the distribution flattens before $k=47$. 
\begin{figure}
    \centering
    \includegraphics[width=0.95\linewidth]{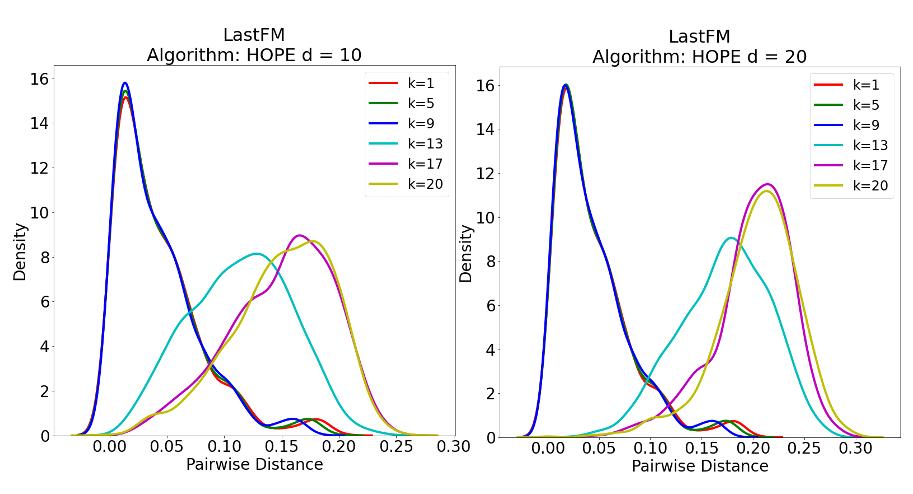}
     \caption{When we embed the LastFM graph with HOPE \cite{HOPE}, we see that the degenerate-core pairwise distribution shifts abruptly between $k=9$ and $k=13$. This holds true for the various embedding dimensions that we tried. For brevity, we only show $d=10$ and $20$. For readability, we have shown the distribution at intervals of $k$.}
    \label{fig:shift-lastfm}
\end{figure}
\begin{figure}
    \centering
    \includegraphics[width=0.95\linewidth]{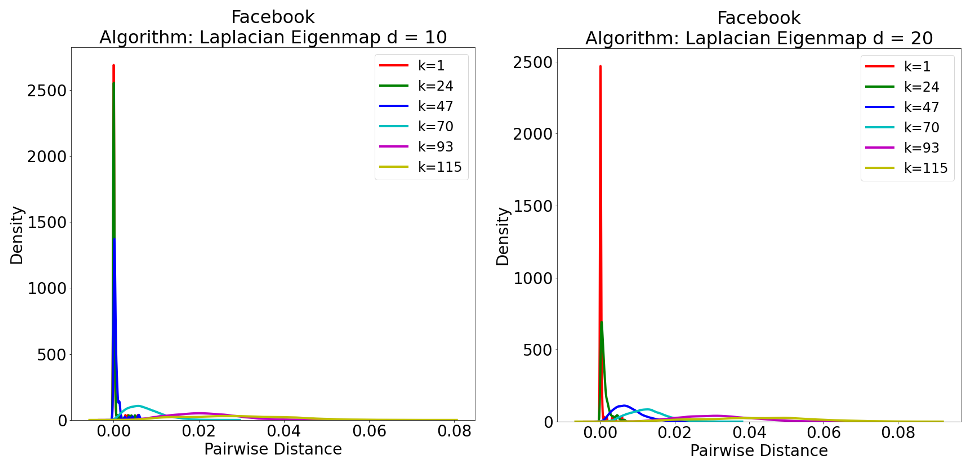}
    \caption{When the degenerate-core pairwise distribution transforms during the shaving process depends on the graph and the embedding dimension. For the Facebook graph, we see the distribution flatten at $k=70$ (cyan) when the embedding dimension $d=10$, but when $d=20$, the distribution flattens earlier in the shaving process at $k=47$ (blue).}
    \label{fig:shift-facebook}
\end{figure}



\vspace*{3pt}
\noindent \textbf{Pattern 2: The degenerate-core embeddings for the ER and BA graphs are stable.}
\vspace*{3pt}

In contrast to the real-world graphs, we found that the embeddings for the ER and BA graphs are stable. Figure \ref{fig:random-stability} shows the degenerate-core pairwise distributions for the Erd\"{o}s-R\'{e}nyi and Barab\'{a}si-Albert graphs. The distributions shown were generated with HOPE \cite{HOPE}, however the pattern holds for Node2Vec \cite{node2vec} and Laplacian Eigenmaps  \cite{lapeigenmap} as well. The stability for these random graphs is likely due to the fact that the degenerate core alone constitutes a large proportion of the entire graph, given the parameters that we selected. For this reason, removing the outer k-shells has less of an impact on the degenerate core.
\begin{figure}
\centering
\begin{subfigure}{0.5\linewidth}
  \centering
  \includegraphics[width=\linewidth]{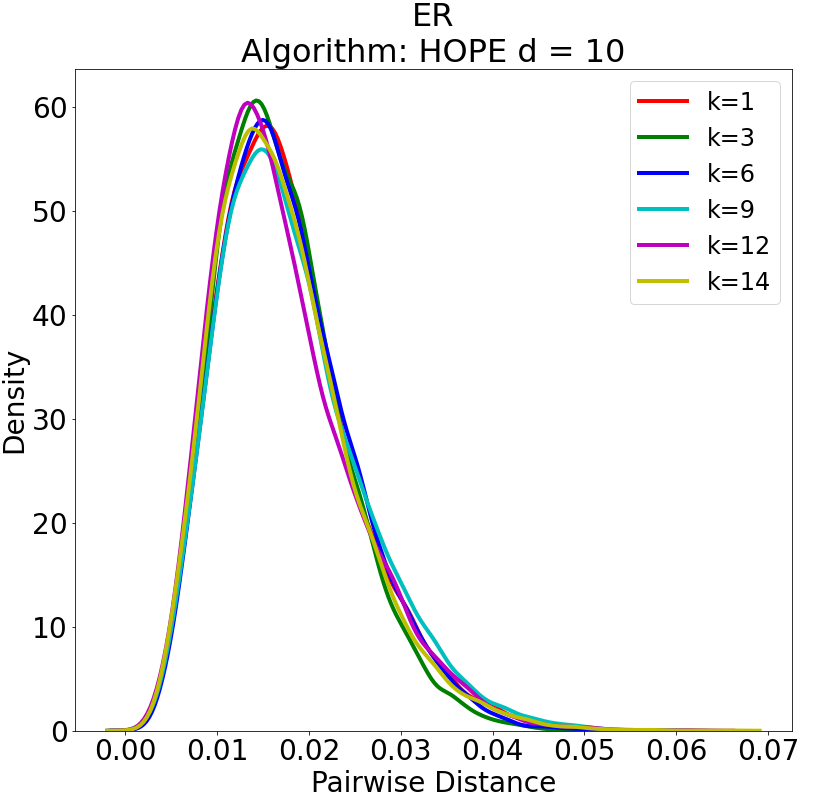}
  \caption{Erd\"{o}s-R\'{e}nyi}
  \label{fig:sub1}
\end{subfigure}%
\begin{subfigure}{0.5\linewidth}
  \centering
  \includegraphics[width=\linewidth]{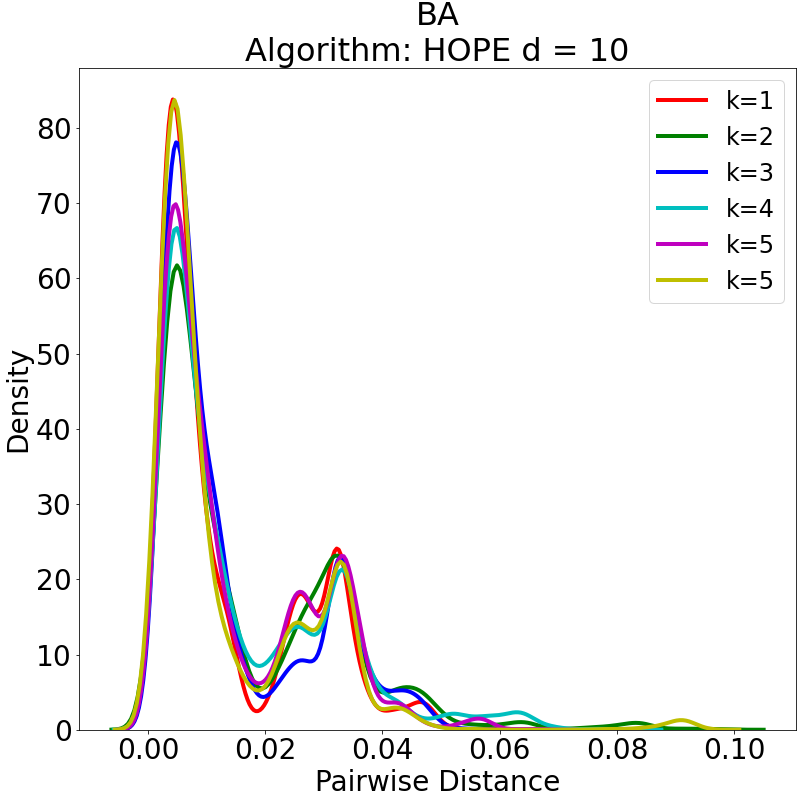}
  \caption{Barab\'{a}si-Albert}
  \label{fig:sub2}
\end{subfigure}
\caption{For Erd\"{o}s-R\'{e}nyi and Barab\'{a}si-Albert graphs, the degenerate core embeddings were stable as the k-shells were removed iteratively. The degenerate-core pairwise distributions above are nearly identical regardless of the subgraph being embedded, and the pattern holds for all six embedding algorithms. This stability is due to the degenerate core constituting a large proportion of the entire graph.}
\label{fig:random-stability}
\end{figure}

\vspace*{3pt}
\noindent \textbf{Pattern 3: As k-shells are removed, the degenerate-core pairwise distribution becomes smoother and unimodal.}
\vspace*{3pt}

We observe that not only does the degenerate-core pairwise distribution shift as k-shells are removed, the distribution also loses modality. In Figure \ref{fig:unimodal}, we show the HOPE embeddings on to the Wikipedia graph. Across all subgraphs, the degenerate-core pairwise distribution is centered close to zero. However, for larger subgraphs, the degenerate-core pairwise distribution is bimodal. When $k=1$, the distribution has a small peak around the distance of $0.6$. For $k>40$, the distribution is nearly unimodal.
\begin{figure}
    \centering
    \includegraphics[scale=0.2]{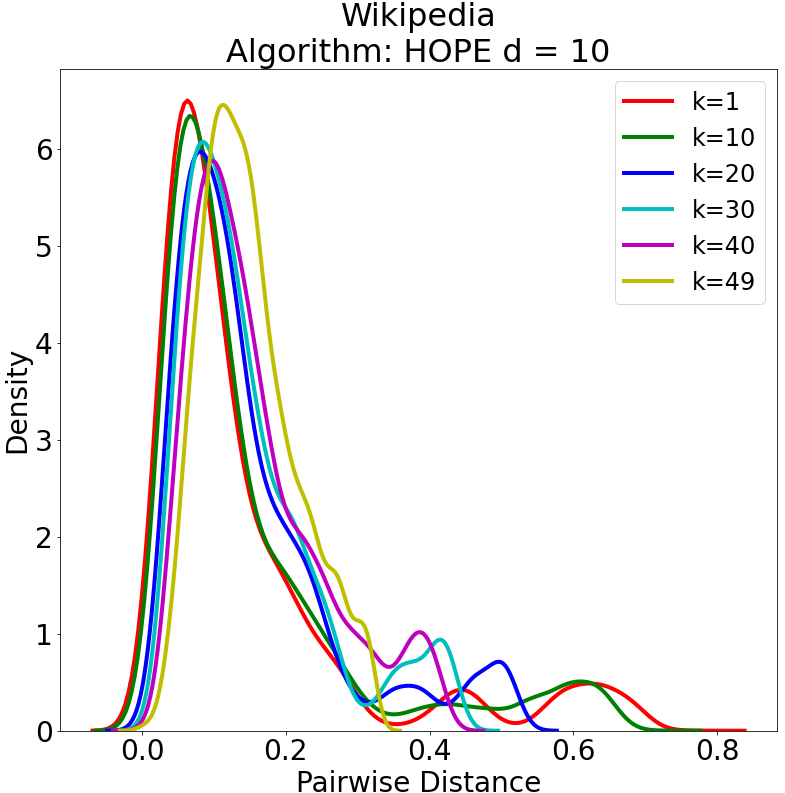}
    \caption{For the Wikipedia graph, as k-shells are removed, the degenerate-core pairwise distribution transforms from being bimodal to unimodal. We conjecture that the loss in modality is due to the loss in community structure.} 
    \label{fig:unimodal}
\end{figure}
\subsection{Significance of the Periphery}
In this section, we examine the causes of the  patterns identified in the previous section. We begin by taking a look at case studies of the Facebook and LastFM graphs. In these two graphs, we see that embedding instability is correlated with subgraph edge-density. We then generalize this result across all graphs by running ordinary least squares regressions that correlate subgraph features with instability. Together, these findings show that the stability patterns found are not simply due to large numbers of nodes being removed from the graph but rather predictable structural changes.

\begin{figure}
    \centering
    \includegraphics[scale=0.25]{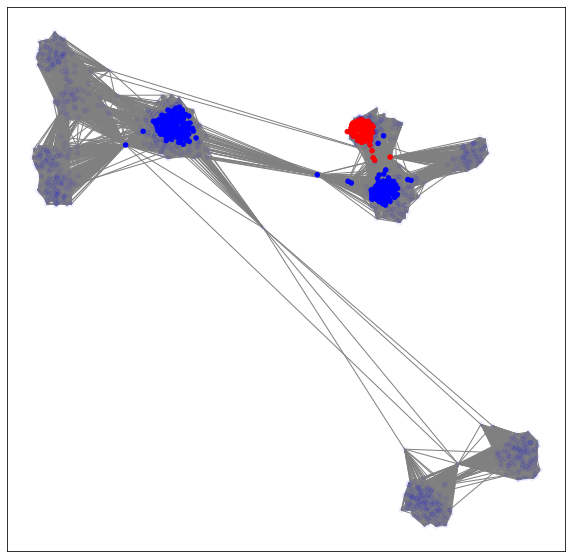} 
    \caption{The above subgraph is the $k=70$ core for the Facebook graph which demonstrates community structure. The red nodes are the degenerate core and all other shells $k=70-114$ are shown in blue where bluer shells create degenerate-core instability once removed. In particular the most salient shell ($k=70$) consists of two clusters.}
    \label{fig:stability_notation}
\end{figure}

\subsubsection{Case Studies: Facebook and LastFM}\label{sec:casestudy}
Figure \ref{fig:casestudy} shows the $k$-core embeddings before and after the point of maximum instability, as defined by Equation \ref{eqn:def-instability}. For both Facebook and LastFM, before the maximum instability shell is removed, the degenerate core, colored in red, is concentrated in a particular region of the embedding space. In the case of Facebook, the $70$-core exhibits core-periphery structure in which the degenerate core is the dense center (see Figure~\ref{fig:stability_notation} for a visualization of this subgraph). However, after one further $k$-shell removal, the degenerate-core embeddings become interspersed with the remaining subgraph. Figure \ref{fig:casestudy} shows the results when embedding with Laplacian Eigenmaps \cite{lapeigenmap} for Facebook and with HOPE \cite{HOPE} for LastFM. We observed this pattern generalized across various embedding dimensions and algorithms.
\begin{figure*}
    \centering
    \includegraphics[width=\linewidth]{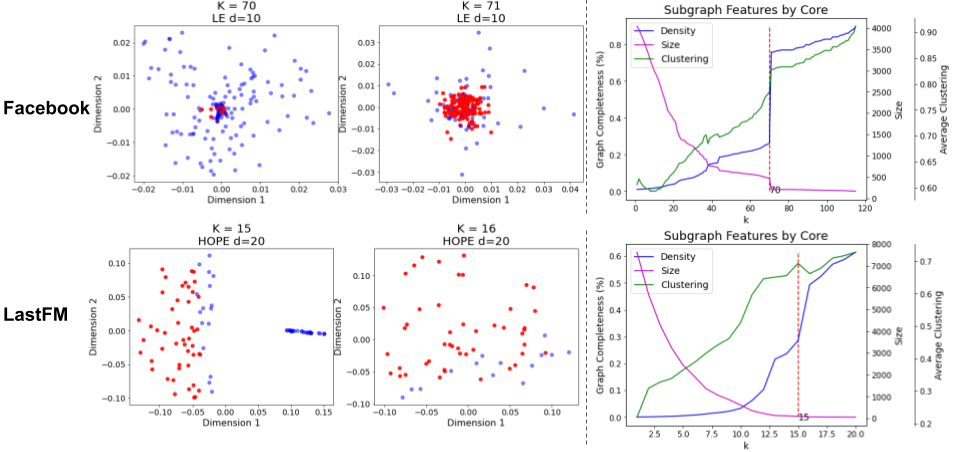}
    \caption{The point of instability of embeddings for both the Facebook and LastFM graphs is correlated with increases in the subgraph density. The left-side scatter plots show the two-dimensional projection of the subgraph embeddings before and after the instability point. In both plots the degenerate core is colored in red. The plots on the right-hand side track subgraph features with each k-shell removal. The point of highest edge density increase (dashed line) is also the instability point.}
    \label{fig:casestudy}
\end{figure*}
The right-hand side of Figure \ref{fig:casestudy} shows the subgraph (edge) density, size, and average clustering coefficient at each $k$-core. In particular, the point of greatest increase in edge density and instability in degenerate-core embedding is denoted by the dashed red line. The plot shows that at the point of greatest embedding instability, the edge density is greatly increasing whereas the size of the subgraph is not dramatically decreasing. In the case of LastFM, the majority of the graph has already been removed. These patterns refute the suggestion that the degenerate-core embeddings are unstable simply because a large number of nodes have been removed from the graph. 

\subsubsection{Regression Analysis}
We model the relationship between $k$-core subgraph features and the corresponding change in the degenerate-core pairwise distribution. As with before, the EMD for the $k$-core is the Earth Mover's Distance between the distribution of degenerate-core pairwise distances when the entire graph is embedded and when only the $k$-core is embedded. Increases in the EMD highlight points of instability. The subgraph features we measure are: the number of nodes (``size"), edge density, average clustering coefficient, and transitivity, which are common features for characterizing subgraphs \cite{subgraph-features}. These features are inputs into the regression model shown in Equation \ref{eqn:emd-regression}, in which we correlate the change in the subgraph features with the change in the pairwise distribution EMD.
\begin{equation}\label{eqn:emd-regression}
    \begin{split}
        \Delta\text{EMD} = \beta_0 &+ \beta_1\Delta_\text{size}\\
        &+ \beta_2\textrm{ }\Delta_\text{edge\_density}\\
        &+ \beta_3\textrm{ }\Delta_\text{clustering\_coefficient}\\
        &+ \beta_4\textrm{ }\Delta_\text{transitivity}
    \end{split}
\end{equation}

The data for the regression model was generated by examining adjacent $k$-cores. For the $k$ and $k-1$ cores of a given graph, we measure the change in the aforementioned subgraph features as well as the change in the EMD, yielding one training datapoint. We repeat this process across all of the graphs and combine the datapoints into a single training dataset. Because the relationship between subgraph features and stability can vary by graph embedding algorithm or the embedding dimension, we ran a regression model for every embedding algorithm and dimension combination. 

Figure \ref{fig:emd-coefficients} shows the results of running the regression in Equation \ref{eqn:emd-regression}. The coefficient values for edge density and size are grouped by embedding dimension, and one value is reported per embedding algorithm. The error bars report $95\%$ confidence intervals. Across all combinations of embedding algorithm and dimension except SDNE with $d=20$, edge density was positively correlated with increases in EMD. Further, the coefficients for edge density were statistically significantly positive in $11$ of the $18$ dimension and algorithm combinations. On the other hand, the coefficients for graph size were always negative which refutes the hypothesis that instability arises simply from removing many nodes from the graph.

\begin{figure*}
    \centering
     \includegraphics[scale=0.45]{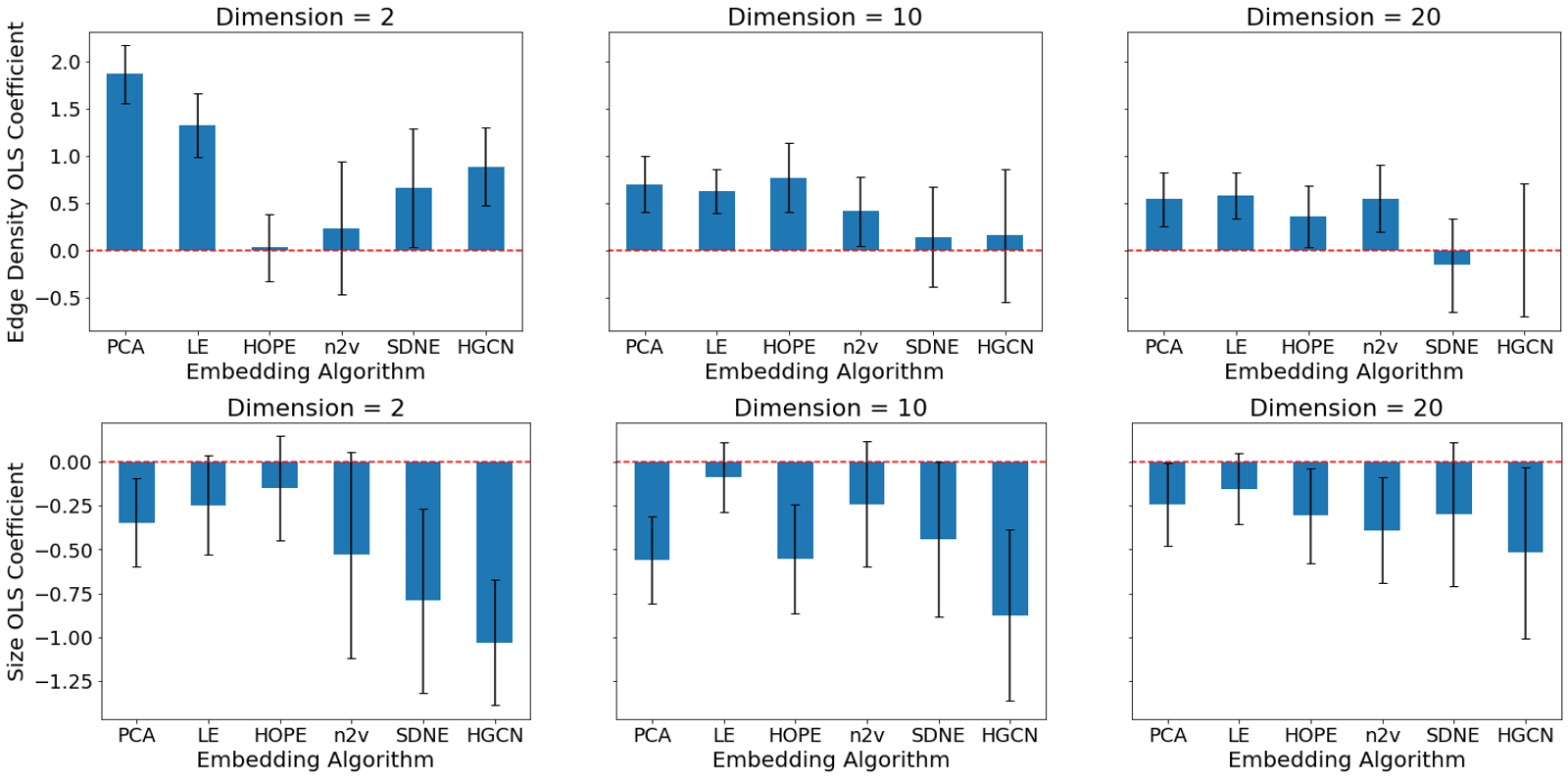}
    \caption{The results of running the regression in Equation \ref{eqn:emd-regression} show that increases in subgraph edge density are correlated with degenerate-core embedding instability. We ran the regression for every combination of embedding dimension and embedding algorithm. The coefficients for edge density (top) and subgraph size (bottom) are shown above. For all dimensions and algorithms except SDNE with $d=20$, the coefficient for edge density was positive. On the other hand, graph size was not positively correlated with instability which refutes the hypothesis that instability arises simply from removing many nodes from the graph. Error bars are provided for statistical significance ($p=0.05$).}
    \label{fig:emd-coefficients}
\end{figure*}

\section{\algName{}: Algorithm for Stable Graph Embeddings}
We propose a graph embedding algorithm \algName{} that produces core-stable embeddings. \algName{} augments any existing graph embedding algorithm by adding an instability penalty to its objective function. Below, we  outline our generic algorithm \algName{} and  show two instantiations of \algName{} by augmenting Laplacian Eigenmaps and LINE. We use the notation introduced in Table \ref{tab:notation}.
\subsection{Generic Algorithm}
\subsubsection{Objective function}
Our objective function consists of two components: the base objective $\mathcal{L}_b$ and an instability penalty $\mathcal{L}_s$. \algName{} minimizes Equation \ref{eqn:base-objective} where $\alpha$ is a regularization hyperparameter.
\begin{equation}\label{eqn:base-objective}
    \pmb{Y}^* = \argmin_{\pmb{Y}\in\mathbb{R}^{n\text{x}d}} \mathcal{L}_b\left(\pmb{Y}, W\right) + \alpha\mathcal{L}_s\left(\pmb{Y}, W, \mathcal{D}\right)
\end{equation}
The instability penalty is high when the degenerate-core embedding is different in the following two cases: (1) the core is embedded in the context of the entire graph and (2) the core is embedded in isolation. We define $\hat{\pmb{Y}}$ as the $\lvert \mathcal{D}\lvert$ x $d$ matrix containing embeddings for the degenerate core when the core is isolated:
\begin{equation}
    \hat{\pmb{Y}}_D = \argmin_{\pmb{Y}\in\mathbb{R}^{\lvert D\lvert\text{x}d}} f\left(W_D, Y\right)
\end{equation}
Now, stability can be defined as the preservation of the first-order proximities between pairs of nodes in the degenerate core, where the first-order proximity $p$ between embedding $\pmb{u}_i$ and $\pmb{u}_j$ is:
\begin{equation}
    p\left(\pmb{u}_i, \pmb{u}_j\right) = \frac{1}{1 + e^{-\pmb{u}_i^T \pmb{u}_j}}
\end{equation}
We define the instability penalty as the sum of squares over all differences in first-order proximities between pairs of degenerate-core nodes:
\begin{equation}\label{eqn:stability-penalty}
    \mathcal{L}_s = \sum_{i,j \in \mathcal{D}} \left\lvert p\left(\pmb{u}_i, \pmb{u}_j\right) - p\left(\pmb{\hat{u}}_i, \pmb{\hat{u}}_j\right) \right\lvert^2
\end{equation}
To optimize \algName{}'s objective function, we perform a batched stochastic gradient descent where each batch is a set of edges. For an edge $\{i, j\}$ where $i$ and $j$ are both in the degenerate core, the gradient for the instability penalty, with constants omitted, is as follows: (The expression for $d\mathcal{L}_b/d\pmb{u}_i$ replaces the last term with $\pmb{u}_i$.) 
\begin{equation}\label{eqn:stability_gradient}
    \frac{d \mathcal{L}_s}{d \pmb{u}_i} = \sigma\left(\pmb{u}_i^T\pmb{u}_j\right)\left[1 - \sigma\left(\pmb{u}_i^T\pmb{u}_j\right)\right]\left[\sigma\left(\pmb{u}_i^T\pmb{u}_j\right) - \sigma\left(\pmb{\widehat{u}}_i^T\pmb{\widehat{u}}_j\right)\right]\pmb{u}_j
\end{equation}
The stability gradient is only used when both $i$ and $j$ are in the degenerate core. For edges outside of the degenerate core, the update rule only utilizes the gradient for the base objective function. In Section \ref{sec:instantiations} we instantiate with Laplacian Eigenmaps and LINE as the base functions and provide the respective gradients there. 

\begin{algorithm}
\caption{\algName{}}\label{alg:cap}
\begin{algorithmic}
\Require $G, W, n_b, \alpha, \eta$
\Ensure $\pmb{Y}^*$
\State $\mathcal{D} \gets \texttt{degenerate\_core}(G)$
\State $\widehat{\pmb{Y}}_D \gets \texttt{base\_embed}(G_D, W_D)$
\State $\pmb{Y} \gets \texttt{base\_embed}(G, W)$
\State $G \gets \texttt{degenerate\_clique}(G, \mathcal{D})$
\State $i \gets 0$
\While{$i < n_b$}
\State $i \gets i + 1$
\State $\text{edges} \gets \texttt{sample\_edges}(G)$
\For{$i,j$ in edges}
\If{$W[i,j] > 0$}
\State $\pmb{Y}[i] \gets \pmb{Y}[i] - \eta \frac{d\mathcal{L}_b}{d\pmb{u}_i}$ \Comment{See Sec. \ref{sec:instantiations}}
\State $\pmb{Y}[j] \gets \pmb{Y}[j] - \eta \frac{d\mathcal{L}_b}{d\pmb{u}_j}$
\EndIf
\If{$i\in\mathcal{D}$ and $j\in\mathcal{D}$}
\State $\pmb{Y}[i] \gets \pmb{Y}[i] - \eta \alpha \frac{d\mathcal{L}_s}{d\pmb{u}_i}$ \Comment{See Eqn. \ref{eqn:stability_gradient}}
\State $\pmb{Y}[j] \gets \pmb{Y}[j] - \eta \alpha \frac{d\mathcal{L}_s}{d\pmb{u}_j}$
\EndIf
\EndFor
\EndWhile
\end{algorithmic}
\end{algorithm}

Algorithm \ref{alg:cap} provides pseudocode for \algName{}. We begin by embedding the degenerate core in isolation ($\widehat{\pmb{Y}}_D$) as well as the entire input graph ($Y$) to initialize the embeddings. Because the instability penalty $\mathcal{L}_s$ sums over all pairs of nodes in the degenerate core, not just connected nodes, the \texttt{degenerate\_clique} method augments the graph by adding an edge between $\{i,j\}\forall i,j\in\mathcal{D}$. These edges are assigned weight zero and when drawn, only the stability update rule is applied and the base update rule is omitted. 
\subsection{Instantiations}\label{sec:instantiations}
\subsubsection{Stable LINE}
When instantiated with LINE (first-proximity) \cite{LINE}, the base loss takes the form:
\begin{equation}
    \mathcal{L}_b = -\sum_{i,j\in E}w_{ij}\log\left(p\left(\pmb{u}_i, \pmb{u}_j\right)\right)
\end{equation}
As is common with LINE implementations, for computational efficiency we utilize negative sampling such that for a sampled edge $i,j$ we minimize the following:
\begin{equation}
    -\log\left(p\left(\pmb{u}_i, \pmb{u}_j\right)\right) - E_{j'\sim P_n}\left[\log\left(p\left(-\pmb{u}_i, \pmb{u}_{j'}\right)\right)\right]
\end{equation}
The gradient for each vertex $\pmb{u}_i, \pmb{u}_j, \pmb{u}_{j'}$ is:
\begin{align}
\begin{split}
    \frac{d\mathcal{L}_b}{d\pmb{u}_i} &= -(1 - \sigma(\pmb{u}_i^T\pmb{u}_j))u_j + \sum_{j'}\sigma(\pmb{u}_i^T\pmb{u}_{j'}))\pmb{u}_{j'}\\
    \frac{d\mathcal{L}_b}{d\pmb{u}_j} &= -(1 - \sigma(\pmb{u}_i^T\pmb{u}_j))u_i\\
    \frac{d\mathcal{L}_b}{d\pmb{u}_{j'}} &= \sigma(\pmb{u}_i^T\pmb{u}_{j'}))\pmb{u}_{i}  
\end{split}
\end{align}
The complexity when instantiated with LINE is $\mathcal{O}\left(ndb + m\right)$ where $b$ is the number of negative samples per edge; $d$ is the number of embedding dimensions; and $n$, $m$ are the number of nodes and edges, respectively. The first term accounts for computing $b$ gradients of size $d$ for $n$ samples and the second term accounts for the overhead needed to setup the edge sampling data structures. Of note, this is the same complexity as LINE itself, so \algName{} does not add to the runtime complexity.

\subsubsection{Stable Laplacian Eigenmaps}
Laplacian Eigenmaps \cite{lapeigenmap} optimizes the following objective function $f$ \cite{lapeigenmap}:
\begin{equation}
    f\left(W, \pmb{Y}\right) = \sum_{i,j} w_{ij}\|\vec{u_i} - \vec{u_j}\|^2    
\end{equation}
However, the optimization is performed over the feasible set $Y^TDY=I$, where $D$ is the diagonal matrix such that $D_{ii}$ is the degree of node $i$. \algName{} initializes with Laplacian Eigenmaps embeddings. For this reason, instead of performing a constrained optimization, we penalize embeddings that deviate from the initial values; adding a deviation penalty proportional to the norm of the difference from the initial embeddings $Y_0$. To balance the orders of magnitude for these two losses, we introduce hyperparameters $\gamma$ and $\beta$. Thus, the base objective when instantiated with Laplacian Eigenmaps is:
\begin{equation}
    \mathcal{L}_b = \gamma f\left(W, \pmb{Y}\right) + \beta\left\|Y - Y_0\right\|^2
\end{equation}
Where the gradient is:
\begin{align}
    \frac{d \mathcal{L}_b}{d \pmb{u}_i} &= \gamma w_{ij}\left(\pmb{u}_i - \pmb{u}_j\right)+ \beta \left(\pmb{u}_i - \pmb{u}_{i0}\right) 
\end{align}
The runtime complexity for the Laplacian Eigmenmap instantiation is $\mathcal{O}\left(nd + m\right)$ because negative sampling is not used. 
\section{Experiments}
Our experiments show that \algName{} produces embeddings that are both core-stable and accurate on link prediction. We have included the details of our experimental setup in Appendix \ref{sec:hyperparam}. For the core-stable results, see Figure \ref{fig:stability_error_distribution}. The figure plots the density distribution of stability errors among all pairs of degenerate nodes, where the stability error for a pair $i,j\in\mathcal{D}$ is defined in Equation \ref{eqn:stability-error}. For both instantiations, the distribution of errors for the \algName{} embeddings centers closer to zero. This implies that \algName{}'s degenerate-core embedding is able to capture the degenerate-core well when all $k$-shells are removed (i.e., in isolation).
\begin{equation} \label{eqn:stability-error}
    \left\lvert p\left(\pmb{u}_i, \pmb{u}_j\right) - p\left(\pmb{\hat{u}}_i, \pmb{\hat{u}}_j\right) \right\lvert^2
\end{equation}
\begin{figure}
    \centering
    \includegraphics[width=\linewidth]{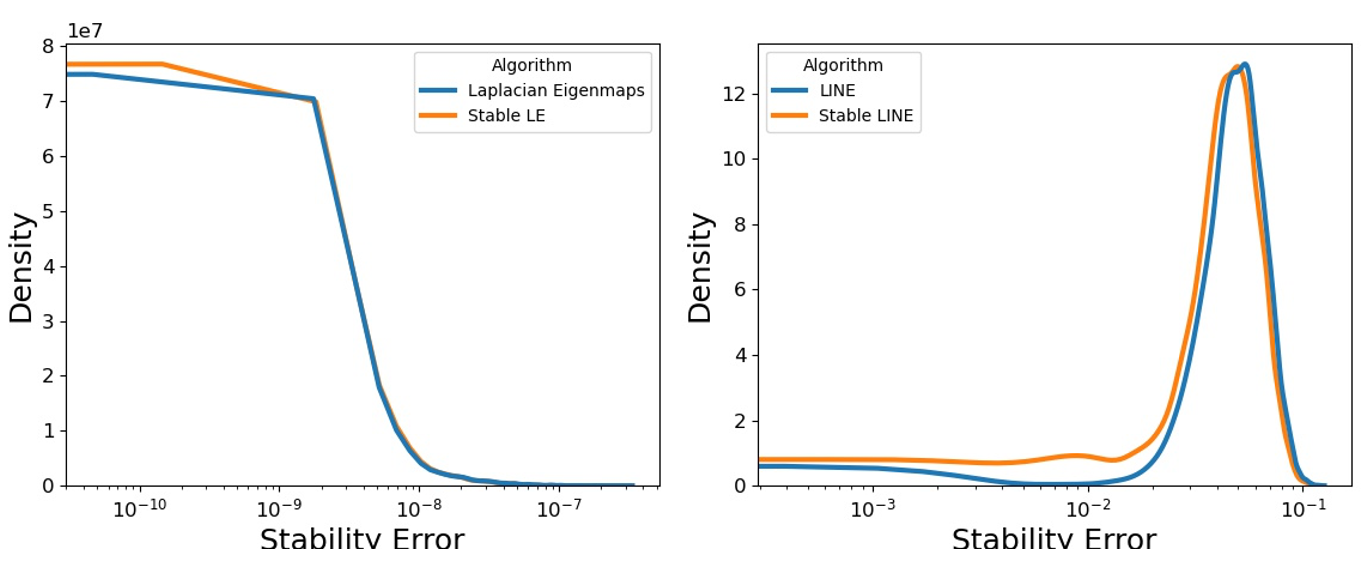}
    \caption{Density distribution of stability errors among all pairs of degenerate nodes. The stability error for a pair of nodes in the degenerate core is defined in Equation \ref{eqn:stability-error}.  
    For instantiations of Laplacian Eigenmaps (left) and LINE (right), the distribution of stability errors is closer to zero for the stable embeddings. That is, \algName{} is able to find stable degenerate-core embeddings.}
    \label{fig:stability_error_distribution}
\end{figure}

Table \ref{tab:link-prediction} lists the results from our link prediction experiments. \algName{}'s graph embeddings preserve and at times improve the link-prediction accuracy of the original embeddings. For LINE, the stable AUC and F1 scores are higher for all graphs even though both algorithms are trained with the same number of iterations. For Laplacian Eigenmaps, the scores are similar for all of the networks, expect for PPI where the stability penalty decreases but the base loss increases.
\begin{table*}[ht]
\centering
\caption{Link prediction performance as measured by FI and AUC for the original (a.k.a. base) embedding and the stable embedding produced by \algName{}. For brevity, we only show results for Laplacian Eigenmaps and LINE. \algName{} not only produces stable degenerate-core embeddings, it also maintains or improves link prediction performance when the base error and the stability penalty decrease together.}
\label{tab:link-prediction}
\begin{tabular}{|l|rrrr|rrrr|}
\toprule
\textbf{Graph} &  \multicolumn{4}{c|}{\textbf{Laplacian Eigenmaps}} & \multicolumn{4}{c|}{\textbf{LINE}}\\
{}   & Original F1 & Stable F1 & Original AUC & Stable AUC & Original F1 & Stable F1 & Original AUC & Stable AUC\\\hline
Facebook & 0.955 & \textbf{0.956} & 0.984 & 0.984 & 0.757 & \textbf{0.815} & 0.843 & \textbf{0.894}\\\hline
AS & \textbf{0.678} & 0.677 & 0.695 & \textbf{0.704} & 0.575 & \textbf{0.619} & 0.610 & \textbf{0.683}\\\hline
PPI & \textbf{0.734} & 0.586 & \textbf{0.807} & 0.629 & 0.535 & \textbf{0.585} & 0.554 & \textbf{0.593}\\\hline
ca-HepTh & \textbf{0.772} & 0.758 & \textbf{0.845} & 0.831 & 0.761 & \textbf{0.794} & 0.840 & \textbf{0.874}\\\hline
LastFM & \textbf{0.864} & 0.850 & \textbf{0.909} & 0.900 & 0.751 & \textbf{0.809} & 0.829 & \textbf{0.883}\\\hline
Wikipedia & \textbf{0.572} & 0.561 & \textbf{0.596} & 0.581 & 0.492 & \textbf{0.512} & 0.490 & \textbf{0.509}\\
\bottomrule
\end{tabular}
\end{table*}
\section{Related Work}
We review related work on $k$-core analysis and the limitations of graph embeddings.
 \textbf{$k$-core Analysis.} $k$-core structure has been important for understanding spreading processes on graphs, in particular identifying the most-influential spreaders \cite{multiplex-core}. Common patterns related to $k$-core structure have been identified such as correlations between a node's degree and coreness (largest $k$ such that the node is in the $k$-core) as well as community structure in dense cores \cite{corescope}. Recent work has also broadened the study of k-cores to consider the addition of ``anchor nodes" that prevent large cascades when individual core nodes are removed \cite{anchors}. Finally it has been shown that not all degenerate cores are equally important; the most salient degenerate cores, called ``true cores" are those that are well interconnected with outer shells \cite{liu_core-like_2015}. 
 \textbf{Limitations of Graph Embeddings.} Practitioners are tasked with choosing from a large selection of algorithms \cite{systematic-comparison} and even once an algorithm has been chosen, hyperparameters such as the embedding dimension can greatly affect performance \cite{gu_principled_2021}. Further, in the case of community detection, expensive algorithms do not always perform traditional community detection algorithms \cite{embedding-communities}. Recent work has also established more theoretical limits to graph embeddings showing that at low dimensions it is impossible for graph embeddings to capture the triangle richness of real-world networks. Stability has also been identified as an issue \cite{schumacher2020effects}, however this study defined stability in a different sense: the consistency of embeddings when the algorithm is re-run multiple times. 




\section{Conclusion}

The degenerate core of a graph is the densest part of that graph. In this work, we examined the stability of embeddings for the nodes in the degenerate core. We defined stability as the property of being resilient to perturbations. We defined perturbations as removing $k$-shells repeatedly from the graph. We observed three patterns of instability across a variety of popular graph embedding algorithms and numerous real-world and synthetic data sets. We also observed a change point in the degenerate-core embedding and showed how it was tied to subgraph density. Subsequently, we introduced \algName{}: an algorithm that takes an existing graph embedding algorithm and adds a stability objective to it. We showed how \algName{} works on two popular graph embedding algorithms and reported experiments that showed the value of \algName{}.

\bibliographystyle{ACM-Reference-Format}
\bibliography{kdd}

\clearpage
\pagebreak
\appendix
\section{Appendex} 

\subsection{Synthetic Networks}\label{sec:synthetic-network}
Below we specify the configurations used to generate the synthetic networks used in our experiments. 
\paragraph{ER} We utilize the \texttt{erdos\_renyi\_graph} generator in the networkx package and call the generator twice with the following parameters: $(n=5000, p=0.002), (n=5000, p=0.004)$.
\paragraph{BA} We utilize the \texttt{barabasi\_albert\_graph} generator in the networkx package and call the generator twice with the following parameters: $(n=5000, m=5), (n=5000, m=10)$.
\paragraph{BTER} We generated the BTER graphs using the Matlab \texttt{FEASTPACK} software package which was downloaded from \url{http://www.sandia.gov/~tgkolda/feastpack/}. The software expects an input degree distribution. For the ``BA (from BA)" graph, we provided the edgelist from the ``BA ($m=5$)" graph detailed above. For the ``BA (Arbitrary)" graph we utilized \texttt{FEASTPACK} to generate an arbitrary degree distribution based on the following specifications: the maximum degree is $< 1000$, the target average degree is $15$, the target maximum clustering coefficient is $0.95$, and the target global clustering coefficient is $0.15$.
\subsection{Laplacian Eigenmaps Clique Theorem}
\label{sec:embed-clique}
\setcounter{theorem}{0}
\begin{theorem}
The random-walk normalized Laplacian for a clique of size $n$ has two eigenvalues: zero, with multiplicity one, and $1 + \frac{1}{n-1}$ with multiplicity $n-1$.
\end{theorem}

\begin{proof}
The random-walk normalized Laplacian is defined as:
\begin{equation}
    L_{rw} = D^{-1}(D - A)
\end{equation}
In the case of a clique, the above matrix evaluates to a matrix in which diagonal entries are $1$ and off-diagonal entries are $-\frac{1}{n-1}$. 

Below, we use three lemmas of eigenvalues, where $\text{eigs}(A)$ refers to the eigenvalues of $A$.
\begin{lemma}\label{lemma:eig-additive}
    Let $A' = A + c*I$, where $I$ is the identity matrix. Then $\text{eigs}(A') = \text{eigs}(A) + c$.
\end{lemma}
\begin{lemma}\label{lemma:eig-multiply}
    Let $A' = c*A$. Then $\text{eigs}(A') = c*\text{eigs}(A)$.
\end{lemma}
\begin{lemma}\label{lemma:eig-complete}
    The eigenvalues of $\mathbbm{1}_n$, an $nxn$ matrix of all ones, are $0$, with multiplicity $n-1$ and $n$, with multiplicity one.
\end{lemma}

Next, we observe that the random-walk normalized Laplacian of a clique can be expressed as:
\begin{equation}
    L_{rw} = \frac{-1}{n-1}*\mathbbm{1} + \left(1 + \frac{1}{n-1}\right)*I
\end{equation}
Then, using Lemmas \ref{lemma:eig-additive} and \ref{lemma:eig-multiply}, the following is true:
\begin{equation}
    \text{eigs}(L_{rw}) = \frac{-1}{n-1}*\text{eigs}(\mathbbm{1}) + \left(1 + \frac{1}{n-1}\right)*I
\end{equation}
Finally, by Lemma \ref{lemma:eig-complete}, the eigenvalues of the normalized Laplacian are $0$ with multiplicity one and $1 + \frac{1}{n-1}$ with multiplicity $n-1$
\end{proof}
\begin{figure}
    \centering
    \includegraphics[width=0.95\linewidth]{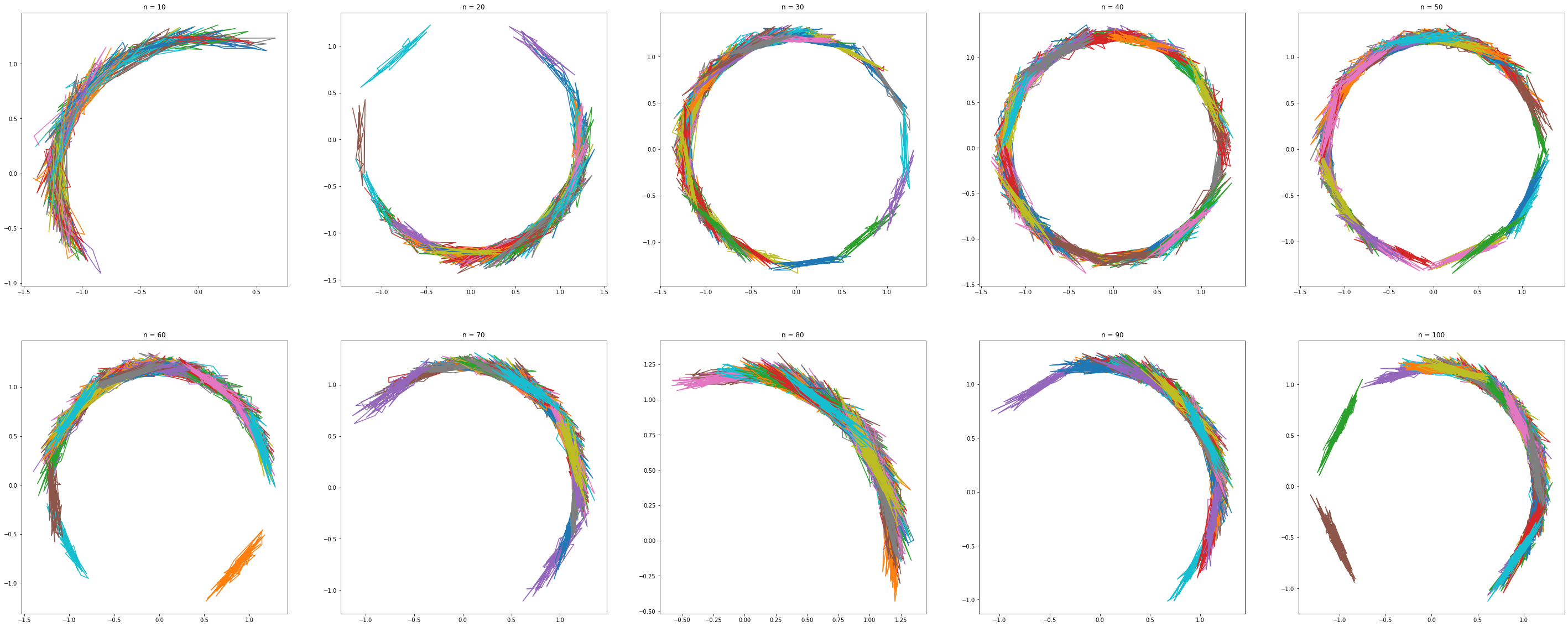}
    \caption{Similar to the Laplacian Eigenmaps, the Node2Vec embeddings for cliques are also unstable. These plots show the Node2Vec embeddings ($d=2$) for cliques of various sizes ($n\in[10,100]$) repeated over multiple iterations. Regardless of the clique sizes chosen, on any given iteration, the embeddings are approximately linear. Across the various iterations (differentiated by color), the embedding instances are collectively circular.}
    \label{fig:n2v_clique}
\end{figure}

\begin{figure}
    \centering
    \includegraphics[width=\linewidth]{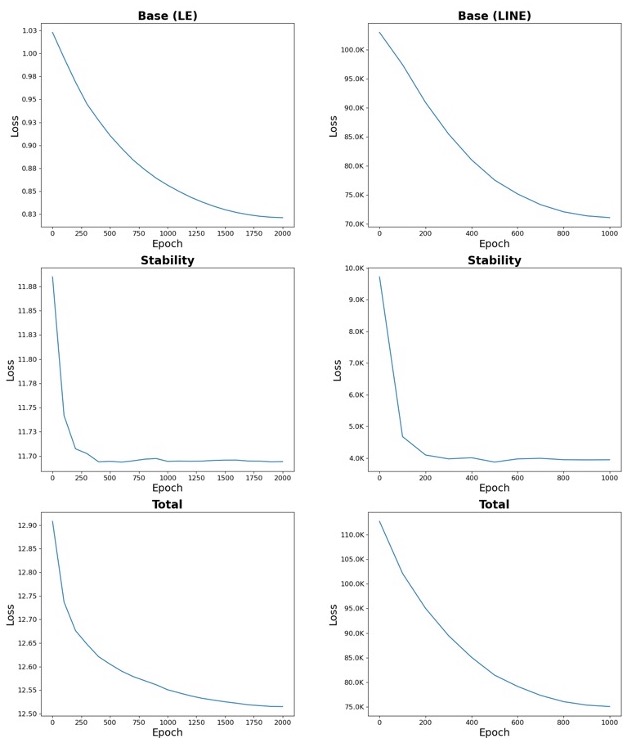}
    \caption{The above plots show the training error when \algName{} is trained on the Facebook graph using Laplacian Eigenmaps (left) and LINE (right) as base functions. In both cases, the base loss and stability loss  decrease over successive training epochs. Furthermore, the stability error decreases immediately while the base error decreases more gradually.}
    \label{fig:training_error}
\end{figure}

\subsection{Node2Vec Embeddings of Cliques}
\label{sec:n2v_clique}
To complement the theoretical analysis of Laplacian Eigenmaps clique embeddings in Section \ref{sec:embed-clique}, the Node2Vec embeddings \cite{node2vec} for cliques are also unstable. When performing a random walk on the clique, at each hop of the walk, every node in the graph is sampled with equal probability (when $p=q=1$, which are the default parameters). 

 We empirically analyze the Node2Vec embeddings for cliques. In the Figure \ref{fig:n2v_clique}, we visualize the Node2Vec embeddings for cliques of various sizes where $n\in[10, 100]$. For each clique, we embed the clique 100 times and visualize the embeddings overlayed on top of each other, differentiated by color. For all of the sizes, we see that the various embeddings are approximately circular; on a particular embedding instance, the embeddings are approximately a tangent along the circle. Similar to the Laplacian Eigenmaps example we can see that for a clique, there are many solutions to the embedding algorithm optimization.


\subsection{Loss Functions}
\label{sec:loss_fns}
Recall that \algName{}'s objective function consists of two components: the base objective $\mathcal{L}_b$ and an instability penalty $\mathcal{L}_s$ (see Equation \ref{eqn:base-objective} for details). We examine the training of \algName{} on the Facebook graph. Figure~\ref{fig:training_error} shows that for both the Laplacian Eigenmaps \cite{lapeigenmap}, the base loss and instability penalty decrease during training. The instability penalty in particular drops precipitously in the first 200 epochs in both cases. The hyperparameters were chosen such that the two losses would be of similar orders of magnitude. 

\subsection{Experimental Setup}
\label{sec:hyperparam}
For our experimental results, we used the following hyperparameter settings: ($\alpha=10^5, \beta=\gamma=0.1$) for Stable Laplacian Eigenmaps and ($\alpha=10$) for Stable LINE. These values were chosen so that the orders of magnitude for $\mathcal{L}_b$ and $\mathcal{L}_s$ are similar. For both instantiations, we use an initial learning rate of $\eta=0.025$ that decreases linearly with each epoch until the rate reaches zero at the final epoch. Furthermore, our link prediction tests withhold 10\% of links for the test set. The labels for link prediction are determined by sorting the cosine similarity scores for all pairs of nodes in the test set and all scores above a set threshold are labels as positive predictions. The threshold is set such that the number for positive predictions matches the number of true positives. 

\subsection{Reproducibility}
We have a GitHub repository for this work available at \href{https://github.com/dliu18/embedding_repo}{\textcolor{blue}{this link}}.


\end{document}